\documentclass{SAGIP}
\usepackage{amsmath,amssymb,amsfonts,bm,amsthm}
\usepackage{multirow}
\usepackage{caption}
\usepackage{subcaption}

\begin{document}

\newcommand{\R}{\mathbb{R}}
\newcommand{\G}{\mathcal{G}}
\newcommand{\V}{\mathcal{V}}
\newcommand{\E}{\mathcal{E}}
\newcommand{\Lcal}{\bm{\mathcal{L}}}
\newcommand{\D}{\bm{\mathcal{D}}}
\newcommand{\A}{\bm{\mathcal{A}}}
\newcommand{\Hcal}{\bm{\mathcal{H}}}
\newcommand{\N}{\mathcal{N}}
\newcommand{\eps}{\bm{\varepsilon}}
\newcommand{\LMI}{\mathcal{LMI}}
\newcommand{\BMI}{\mathcal{BMI}}
\newcommand{\0}{\bm 0}
\newcommand{\bmzeta}{\bm{\zeta}}
\newcommand{\bmPsi}{\bm{\Psi}}
\newcommand{\VMA}{\text{VMA}}

\newtheorem{lemma}{Lemma}
\newtheorem{theorem}{Theorem}
\newtheorem{remark}{Remark}

\title{Dynamic event-triggered control for multi-agent systems with adjustable inter-event time: a moving average approach}

\def\shorttitle{DETC for MAS: a Moving Average Approach}

\author{Zeyuan Wang\inst{1}, Mohammed Chadli\inst{2}}

\institute{
University Paris-Saclay, Univ Evry, IBISC, 91020 Evry, France
 \\
\email{zeyuan.wang@universite-paris-saclay.fr}
\and
University Paris-Saclay, Univ Evry, IBISC, 91020 Evry, France
\\
\email{mohammed.chadli@univ-evry.fr}
}

 \begin{figure*}[!]
\begin{center} 
\centering
\label{fig:SusF}
\end{center}
\vspace{-1.5cm}
\end{figure*} 

\maketitle
\thispagestyle{empty}

\keywords{Multi-agent systems, Leader-following consensus, Dynamic event-triggered control, Moving average.}

\section{Introduction}
The research and application of multi-agent systems (MASs) originated in the 1980s and gained widespread development in the mid-1990s. In recent years, the development of technologies such as unmanned aerial vehicle coordination, underwater cooperation, and robot formation control has brought consensus issues in MASs to the forefront of global research. One of the particularly intriguing topics is leader-following consensus, also known as model reference consensus, where a group of agents needs to achieve consensus with the leader agents. This topic has been extensively explored for both linear and nonlinear MASs, and for MASs comprising homogeneous or heterogeneous agents under different scenarios.

Unlike the traditional time-sampling control, event-triggered (ET) control has the advantage of triggering sampling and communication dynamically and adaptively based on the state of the system, resulting in significantly reduced network load and energy consumption. Numerous studies have explored the application of this control method in MASs, such as in \cite{Trejo_Chadli22}. However, designing an efficient ET rule with adjustable inter-event time (IET) while avoiding the Zeno effect is still challenging. The dynamic ET mechanism (DETM) \cite{Yi_Liu19} has brought new ideas by introducing auxiliary dynamic variables (ADVs) to relax the Lyapunov function, thereby increasing IET. However, rigorous proof of the Zeno effect without introducing a dwell-time remains still difficult. A recent study \cite{Berneburg20} proposed a new DETM to address these issues, but its application in generic linear MAS may not guarantee global consensus \cite{Wu_Mao22}.

This paper is inspired by recent research on DETM using clock-like auxiliary dynamic variables (ADVs) \cite{Berneburg20, Chu_Huang20, Wu_Mao22}. However, a discontinuous Lyapunov function was chosen in \cite{Wu_Mao22}, leading to non-consensus results. Moreover, the DETM only guarantees stability between two consecutive events but not globally. Based on the above discussion, this paper proposes an improved DETMs to address the shortcomings of \cite{Wu_Mao22} for the leader-following problem. Our approach enforces a global decay rate of the discontinuous Lyapunov function using a moving average method, enabling IET tuning by varying the upper bound of ADVs and the length of the moving horizon. Numerical examples are given to validate the proposed DETMs, demonstrate the parameter's effect, and compare with the state-of-the-art results of \cite{Wu_Mao22}.

\section{Results}
\subsection{Improved DETM based on moving average approach}
Consider a linear MAS with $N$ follower agents and one leader represented by
\begin{equation} \label{equ:MAS model}
    \begin{aligned}
        \dot{\bm{x}}_i(t) &= \bm{A}\bm{x}_i(t)+\bm{B}\bm{u}_i(t), i=1,...,N , \quad \dot{\bm{x}}_0(t) &= \bm{A}\bm{x}_0(t)
    \end{aligned}
\end{equation}

where $\bm{x}_i(t)$ is the follower agents' state, $\bm{x}_0(t)$ is the leader's state and $\bm{\eps}_i(t) = \bm{x}_i(t)-\bm{x}_0(t)$ is the consensus error of agent $i$. 

The proposed control input $\bm u_i(t)$ of agent $i$ is defined as $\bm u_i(t) = \bm K \bm z_i(t) \text{ and }  \bm z_i(t) = \sum_{j=1}^N a_{ij}(\hat {\bm x}_j(t)- \hat {\bm x}_i(t)) + d_i(\bm x_0(t)-\hat{\bm  x}_i(t)) $. $\hat {\bm x}_i(t)$ is a kind of model-based estimation defined as:
$ \hat {\bm x}_i(t) = \bm x_i(t_k^i)e^{\bm A(t-t_k^i)}, t \in [t_k^i, t_{k+1}^i) $, where $\bm x_i(t_k^i)$ is the value of $\bm x_i(t)$ at the last triggering moment $t_k^i$, and $t_{k+1}^i$ is defined by the ET rule given in the following theorem. Define $\bm e_i(t) = \hat {\bm x}_i(t) - \bm x_i(t)$ and $\bm e(t) = [\bm e_1^T(t),...,\bm e_N^T(t)]^T$. 

The communication topology of $N$ follower agents is represented by an undirected weighted graph $\G=(\V, \E)$. Follower agents and the leader agent are represented as vertices $v_i, i \in \{1,...,N\}$ and $v_0$, respectively. Define $\bar \G = (\bar \V, \bar \E)$ as the augmented graph of $\G$, with $\bar \V = \V \cup \{v_0\}$ and $\bar \E = \E \cup \Delta \E$, where 
$(v_i, v_0) \in \Delta \E$ if agent $i$ is connected to the leader.

We propose the following theorem of improved dynamic event-triggered control which enhances the stability of the MAS compared to the study \cite{Wu_Mao22}, with designable inter-event time. 
\begin{theorem}\label{lemma: etm in soto} 
    Assume that $(\bm A,\bm B)$ is stabilizable and the communication topology graph between $N$ follower agents is weighted, undirected and fixed. If the augmented graph $\bar \G$ contains a spanning tree with the leader agent being its root, the MAS (\ref{equ:MAS model}) will reach leader-following consensus asymptotically with the event-triggered rule defined as:
    \begin{equation} \label{equ: dot theta ref}
        t_{k+1}^i  \triangleq 
        \left\{
        \begin{aligned}
            &\inf \left\{t>t_{k}^i \mid \theta_i(t) \leq 0\right\}, \\
            &\theta_i(t_{k}^i) = \bar{\theta_i}(t)>0
        \end{aligned}
        \right.
        ,
        \dot{\theta}_i(t) \triangleq 
        \begin{cases}
        \omega_i(t)-\tau_i & \text { if } \|\bm e_i(t)\| \neq 0 \text { or } \|\bm z_i(t)\| \neq 0 \\
        -\tau_i & \text { otherwise } 
        \end{cases}
    \end{equation}
    where $\tau_i > 0$ is a constant scalar and $\omega_i(t) \triangleq  -((\alpha \theta_i+2 \delta \beta+(\theta_i-1)^2 ) \|\bm e_i\|^2 + (2 \varepsilon \theta_i+\beta^2) \|\bm z_i\|^2 ) / ( \eta \|\bm e_i\|^2+\|\bm z_i\|^2 ) $, and  $\bar \theta_i(t)$ is defined as 
    \begin{equation} \label{equ:bar theta adaptive}\bar \theta_i(t_k^i) = 
    \begin{cases}
            \frac{F}{\|\bm z_i(t_{k}^i)\|^2}, & \text{if } \|\bm z_i(t_{k}^i)\|^2 \geq \frac{F}{\bar \theta_i(0)} \\
            \bar \theta_i(0), & \text{otherwise}
    \end{cases}
    \end{equation}
where $F$ is defined as  $F = \eps_i^T(t_{k-\ell}^i)\bm P\eps_i(t_{k-\ell}^i)e^{-\rho(t_k^i-t_{k-1}^i)} - \eps_i^T(t_k^i)\bm P\eps_i(t_k^i) 
    + \bar \theta_i(t_{k-\ell}^i) e^{-\rho(t_k^i-t_{k-1}^i)}$ $\|\bm z_i(t_{k-\ell}^i)\|^2$, $\alpha$, $\delta$, $\beta$, $\eta$, $\varepsilon$ are designed parameters, $\ell \in \mathbb{N}^+$ is the moving average steps, and $\rho>0$.
\end{theorem}

\begin{proof}
    Due to page limits, we can only provide a brief proof and key ideas. The adaptive rule of $\bar \theta_i(t)$ in Equation \ref{equ:bar theta adaptive} is based on a moving average of the Lyapunov function defined as $V = \sum_{i=1}^N V_i$ with $V_i = \eps_i^T \bm P \eps_i+ \theta_i\bm e_i^T\bm P\bm e_i + \theta_i \|\bm z_i\|^2$. The rule (\ref{equ: dot theta ref}) enables $V$ decreasing between arbitrary two consecutive events and the rule (\ref{equ:bar theta adaptive}) guarantees that the sequence $\{V_i(t_k^i)\}_{k \in \mathbb{N}}$ converges to zero. Substituting (\ref{equ:bar theta adaptive}) in $V_i$ gives $ V_i(t_{k}^i) \leq V_i(t_{k-\ell}^i)e^{-\rho(t_{k+1}^i-t_k^i)}$, 
    which implies that $V_i(t_{k}^i) < V_i(t_{k-\ell}^i)$ thus both the sequence $\{V_i(t_k^i)\}_{k \in \mathbb{N}}$ and $V(t)$ converge to 0.
    By adding $\sum_{n=k-\ell+1}^{k-1}V_i(t_n^i) $ on the left and right of $V_i(t_{k}^i) < V_i(t_{k-\ell}^i)$, and divide both sides by $\ell$, we have $\VMA(t_k^i, \ell) <  \VMA(t_{k-1}^i, \ell)$, where $\VMA(t_k^i, \ell)$ is defined as $\VMA(t_k^i, \ell) = \left[V_i(t_k^i) + V_i(t_{k-1}^i) + ... + V_i(t_{k-\ell+1}^i)\right] / \ell$. 
$\VMA(t_k^i, \ell)$ is the $\ell$-steps moving average of $\{V_i(t_k^i)\}_{k \in \mathbb{N}}$. It is noticeable that the sequence $\{V_i(t_k^i)\}_{k \in \mathbb{N}}$ may increase locally, but in a global view (represented by $\VMA(t_k^i, \ell)$), it keeps a decreasing trend with a moving horizon of length $\ell$. In addition, this theorem excludes the Zeno behavior, which can be proved by finding a lower bound of $\omega_i$ and calculating the time required for $\theta_i$ decreasing from $\bar \theta_i$ to 0.
\end{proof}

\begin{remark}
This proposed strategy provides more tolerance to small perturbations of $V$ and eventually increases the inter-event time as long as possible while possibly sacrificing some performance temporarily. Compared to the original DETM in \cite{Wu_Mao22}, the instability problem is well addressed by varying $\bar \theta_i$ in an adaptive manner.
\end{remark}

\subsection{Numerical Example}
We demonstrate the effectiveness of the proposed strategy by the following example. Consider a generic linear MAS associated with the topology in Figure (\ref{fig:topology graph}). Figure (\ref{fig:consensus}) illustrate the trajectories of consensus error under the proposed DETM with $\bar \theta_i(0)=5000$ and $\ell=50$, which shows an asymptotic convergence.
Table~\ref{tab:IET comparison} presents how varying moving average steps $\ell$ and $\bar \theta_i(0)$ affects the performance. The inter-event time could be enlarged significantly by increasing $\ell$ or $\bar \theta_i(0)$, e.g., from 1.69 ms ($\bar \theta_i(0)=1000$, $\ell=3$) to 32.83 ms ($\bar \theta_i(0)=13000$, $\ell=100$). Compared to the existing work \cite{Wu_Mao22}, the global convergence is strictly guaranteed, and the IET is much longer than the classic static event-triggered mechanism (SETM).

\vspace{0.5cm}

\begin{minipage}{\textwidth}
  \begin{minipage}[b]{0.22\textwidth}
    \centering
    \includegraphics[height=1.7cm,clip=true]{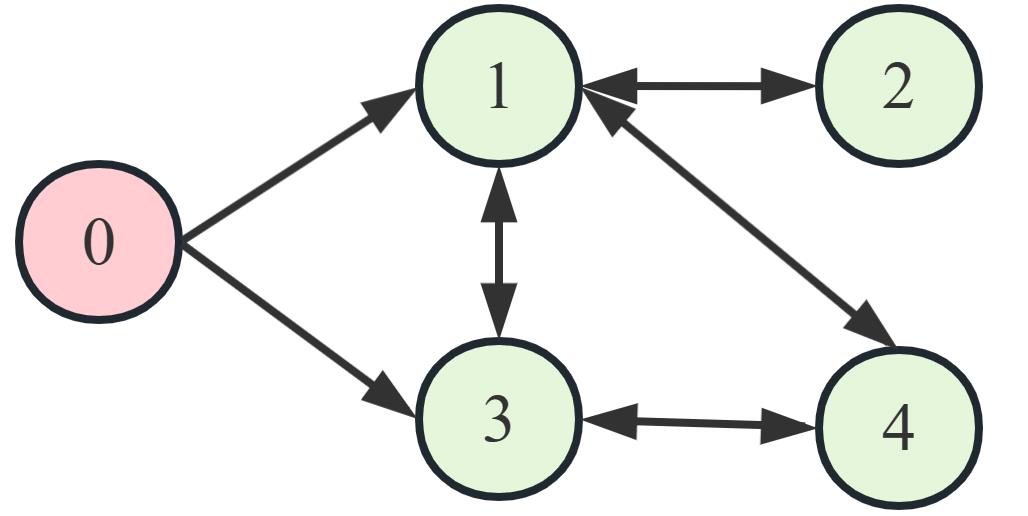}
    \vspace{0.8cm}
    \captionof{figure}{Communication topology}
    \label{fig:topology graph}
  \end{minipage}
  \hfill
    \begin{minipage}[b]{0.3\textwidth}
    \centering
    \includegraphics[height=3.5cm,clip=true]{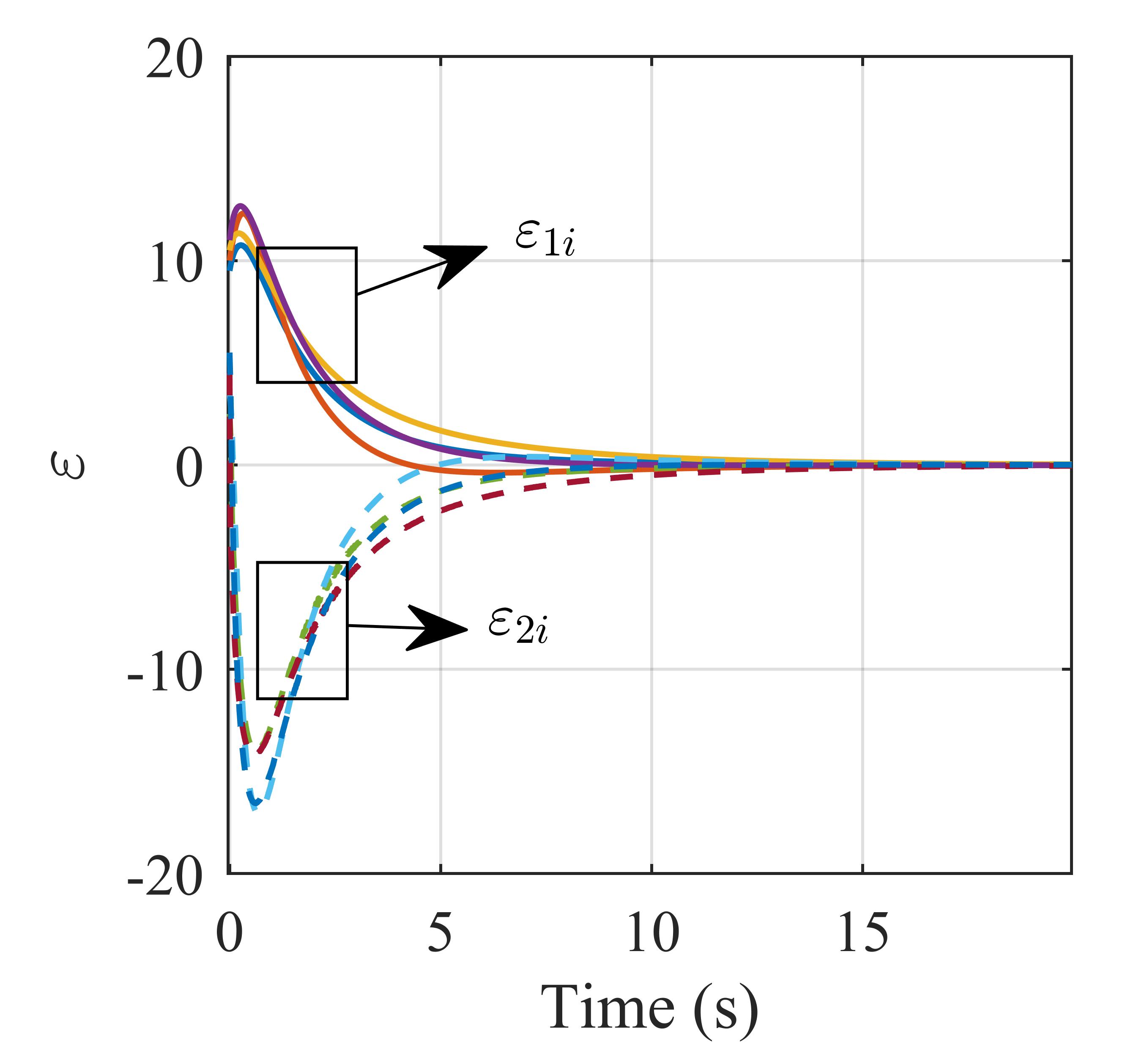}
    \captionof{figure}{Trajectory of consensus error}
    \label{fig:consensus}
  \end{minipage}
  \hfill
  \begin{minipage}[b]{0.4\textwidth}
    \centering
        \begin{tabular}{lccccc}
        \hline
        \multicolumn{1}{l}{$\overline \theta_i(0)$} & 1000   & 5000   & 13000    \\ \hline
        $\ell=3 $        & 1.69     & 3.49     & 8.07     \\
        $\ell=10 $       & 3.79     & 11.72    & 16.51    \\
        $\ell=50 $       & 4.40     & 20.09    & 30.04   \\
        $\ell=100 $      & 4.63     & 24.61    & 32.83   \\ 
        SETM & \multicolumn{3}{c}{0.82} \\
        Ref \cite{Wu_Mao22} & \multicolumn{3}{c}{unguaranteed consensus} \\   
        \hline
        \end{tabular}
      \captionof{table}{IET (in ms) using different $\ell$ and $\bar \theta_i(0)$, compared to other strategies}
      \label{tab:IET comparison}
    \end{minipage}
\end{minipage}

\section{Conclusion}

This paper investigates the leader-following consensus problem of MASs and proposes an improved DETM strategy that ensures asymptotic consensus and prevents Zeno behavior. By using the moving average method to enforce a global convergence, a successful solution to the instability problem is proposed. The proposed result shows that the performance in terms of inter-event time has been significantly improved compared to existing results and the static event-triggered control. Further works will aim at observer-based consensus control.

\bibliographystyle{plain}


\end{document}